\documentclass[12pt]{article}
\usepackage{enumitem}
\usepackage{amsthm}
\usepackage{fancyhdr}
\usepackage{mathrsfs}
\usepackage{natbib}
\usepackage[margin=1in]{geometry}
\usepackage{color}
\usepackage{multirow,array}
\usepackage{amsmath,amsthm,amssymb}
{
\theoremstyle{plain}
\newtheorem{theorem}{Theorem}[section]  \newtheorem{lemma}{Lemma}[section]
\newtheorem{assumption}{Assumption}

\newtheorem{example}{Example}

  }
\usepackage[colorlinks = true,
            linkcolor = blue,
            urlcolor  = blue,
            citecolor = blue,
            anchorcolor = blue]{hyperref}

\usepackage{verbatim}
\usepackage{tikz}
\usepackage{mathtools}

\usepackage{appendix}
\usepackage{float}
\usepackage{graphicx}
\usepackage{caption}
\usepackage{subcaption}
\usepackage[utf8]{inputenc}
\usepackage[T1]{fontenc}    
\usepackage{hyperref}       
\usepackage{url}            
\usepackage{booktabs}       
\usepackage{amsfonts}       
\usepackage{nicefrac}       
\usepackage{microtype}      

\usepackage{dsfont}
\usepackage{indentfirst}

\let\emptyset\varnothing

\numberwithin{equation}{section}

\title{Inference on Partially Identified Parameters with Separable Nuisance Parameters: a Two-Stage Method
\thanks{I am grateful to Hiroaki Kaido, Ivan Fernandez-Val, and Jean-Jacques Forneron for their guidance and support. I also thank participants of Boston University econometrics workshop for their helpful comments.}
}
\author{Xunkang Tian\footnote{European Research University. Email: \href{mailto:xunkang.tian@eruni.org}{xunkang.tian@eruni.org}}}
\date{\today}

\begin{document}

\maketitle

\begin{abstract}
This paper develops a two-stage method for inference on partially identified parameters in moment inequality models with separable nuisance parameters. In the first stage, the nuisance parameters are estimated separately, and in the second stage, the identified set for the parameters of interest is constructed using a refined chi-squared test with variance correction that accounts for the first-stage estimation error. We establish the asymptotic validity of the proposed method under mild conditions and characterize its finite-sample properties. The method is broadly applicable to models where direct elimination of nuisance parameters is difficult or introduces conservativeness. Its practical performance is illustrated through an application: structural estimation of entry and exit costs in the U.S. vehicle market based on \cite{wollmann2018trucks}.
\end{abstract}

\begin{center}
{\small \textbf{Keywords:} Moment Inequalities, Nuisance Parameters, Partial Identification}
\end{center}

\clearpage

\section{Introduction}

In a moment inequalities system, researchers often encounter parameters of interest that are not point identified. The literature has developed a variety of methods to construct confidence intervals for these partially identified parameters. For instance, \cite{chernozhukov2007parameter} proposes a criterion function and characterizes the identified set using the set of minimizers. \cite{rosen2008confidence} constructs a quadratic-form test statistic, deriving its asymptotic behavior, which follows a chi-square distribution. After determining the critical value for the parameters of interest, the confidence set can be obtained through test inversion. \cite{andrews2010inference} introduce their test and confidence set, building upon the generalized moment selection method.
Further contributions to the literature include the work of \cite{canay2010inference}, which focuses on empirical likelihood inference for partially identified models and establishes the validity of the bootstrap for EL-based test statistics. Additionally, \cite{romano2010inference} develops a general framework for conducting inference on the identified set in partially identified econometric models that is uniformly valid across a broad class of models, including those with moment inequalities. These studies significantly expand the toolbox of techniques available for constructing confidence sets in partially identified moment inequality models, providing researchers with a diverse array of approaches to tackle complex estimation problems.

Researchers may sometimes only be interested in specific parameters within the moment inequalities system, focusing on "subvector inference." In these cases, handling nuisance parameters and eliminating unnecessary parameters can be a challenging task. Traditional methods, such as subvector projection and subsampling inference, often exhibit conservatism and limited asymptotic power. Fortunately, recent literature has introduced several progressive improvements in this area.
\cite{bugni2017inference} introduces a bootstrap-based inference method specifically designed for subvector inference, which addresses some of the shortcomings of traditional approaches. Similarly, \cite{chen2018monte} proposes an algorithm that constructs confidence sets for subvectors based on Monte Carlo simulations, providing a more accurate estimation of the parameter space. Additionally, \cite{kaido2019confidence} presents a bootstrap-based calibrated projection procedure for constructing confidence sets, tailored for a single component or a smooth function of the parameter vector.
These methods generally control asymptotic size uniformly across a broad class of data generating processes, offering more reliable and effective techniques for subvector inference in moment inequalities systems compared to their traditional counterparts.

In the literature, there are studies that focus on inference for partially identified parameters in specific structures, which are frequently encountered in empirical applications. While these studies consider a limited class of data generating processes, they offer improved efficiency and better asymptotic convergence rates compared to more general methods.
\cite{andrews2019inference} examines moment inequalities that are linear with respect to the nuisance parameters. The authors develop a least favorable test that addresses the nuisance parameter and reduces conservativeness. 
\cite{cox2023simple} considers a linear moment inequalities system and eliminates the nuisance parameter using a linear algebra method developed by \cite{kohler1967projections}, combined with vertex enumeration. Subsequently, they develop the refined chi-squared (RCC) test to estimate the parameters of interest. 
By concentrating on specific structures, these studies offer valuable insights into the inference of partially identified parameters in commonly encountered situations. These tailored approaches can lead to computational savings, making them particularly useful for empirical applications.


In this paper, we address a specific structure in which the nuisance parameters are separable from the parameters of interest. In this structure, nuisance parameters are contained within a single function that is unrelated to the parameters of interest, and the moment inequalities system is linear with respect to this function. This formulation goes beyond the linearity restriction of nuisance parameters found in some previous literature, allowing for a wider range of application scenarios.

Under certain assumptions and circumstances, separable nuisance parameters can be eliminated through direct elimination, as demonstrated by \cite{cox2023simple}. However, this approach requires considerable computational effort, particularly as the complexity of the vertex enumeration process grows rapidly with the dimensions of moment inequalities. Furthermore, this method introduces conservativeness into the final results.

To overcome these limitations, we develop a two-stage method for constructing a confidence set for the parameters of interest. The first stage involves estimating the nuisance parameters. This stage is compatible with a wide range of methods, such as Generalized Method of Moments (GMM) or Maximum Likelihood, provided that the first-stage estimator satisfies a specific assumption detailed later in the paper.

In the second stage, we base our method on the refined chi-squared  test developed by \cite{cox2023simple}, incorporating necessary adjustments. We construct a test statistic and derive its asymptotic properties, enabling the calculation of a confidence region for the parameters of interest through test inversion. Although our method requires an additional stage of estimation for the nuisance parameters, it produces a confidence set with the correct asymptotic rate of coverage and avoids the conservativeness issue associated with the direct elimination of nuisance parameters. 
This approach, therefore, offers a reliable alternative for estimating partially identified parameters in moment inequalities systems with separable nuisance parameters, at the cost of assuming the nuisance parameters can be estimated separately.

The rest of the paper is organized as follows: Section \ref{ch:method} sets up the moment inequalities and describes the method of estimation, Section \ref{ch:asymptotics} discusses the asymptotic and finite-sample properties of the approach, 
Section \ref{sec:empirical} presents an empirical application with U.S. vehicle market,
and Section \ref{ch:conclusion} concludes this paper.

\section{Set-up and Inference Methods} \label{ch:method}

\subsection{Set-up and motivating examples}\label{exampleref1}

In this section, we describe the method of inference. Let $W \in \mathbb{R}^{d_W}$ be an observable random vector with distribution $F$. 
We have $n$ observations of $W$, denoted as $\{ W_i \}_{i=1}^n$. Let $\theta \in \Theta \subseteq \mathbb{R}^{d_\theta}$ and $\delta \in \Delta$ be unknown parameter vectors, where $\Delta$ is a subset of a vector space which can be finite or infinite dimensional. Suppose $\theta$ is our parameter of interest, and $\delta$ is a nuisance parameter.


Let $M : \mathbb{R}^{d_W} \times \Theta \to \mathbb{R}^{d_M}$ and $N : \mathbb{R}^{d_W} \times \Delta \to \mathbb{R}^{d_N}$ be measurable functions. Let $B : \Theta \to \mathbb{R}^{k \times d_M}$, $C : \Theta \to \mathbb{R}^{k \times d_N}$, and $\rho : \Theta \to \mathbb{R}^{k \times 1}$ be matrix-valued measurable functions, respectively. Consider the following moment inequality restrictions:
\begin{equation}\label{a}
\mathbb E_F[ B(\theta) M(W, \theta) - C(\theta) N(W, \theta, \delta) ] \leq \rho(\theta).
\end{equation}
Note that Equation \eqref{a} can represent any general form of moment inequalities, expressed as 
\[
E_F[m(W, \theta, \delta)] \geq 0,
\] 
by setting the following components: \( B(\theta) = 0 \), \( M(W, \theta) = 0 \), \( C(\theta) = 1 \), \( N(W, \theta, \delta) = m(W, \theta, \delta) \), and \( \rho(\theta) = 0 \). We adopt the specific structure in Equation \eqref{a} to achieve a clear separation between the nuisance parameter and the moment inequalities. In this form, the nuisance parameter \( \delta \) only appears through \( N(W, \theta, \delta) \), ensuring that \( N \) remains non-divisible. This separation is aimed at simplifying the computational procedures in subsequent sections. 
Therefore, we still use "separable" to describe our moment inequalities even though the separation between parameters of interests and nuisance parameters is in fact not required. 
We bring two economic examples matching with this moment inequalities structure. 


\newcounter{example1}
\setcounter{example1}{\value{example}}

\begin{example}[US Vehicle Market]\label{example1}
The first example is from \cite{wollmann2018trucks}, 
which builds a structural model of US commercial vehicle markets. 
Suppose the utility function for consumer $i$ purchasing vehicle type $j$ is $U(p_j,x_j,\xi_j,\varepsilon_{ij};\delta)$, 
where $x_j$ is a $d_X$-vector of observable product characteristics, $p_j$ is the price of product $j$, $\xi_j$ represents unobservable product characteristics, and $F$ is their joint distribution. 
$\varepsilon_{ij}$ is an idiosyncratic error term that follows an independent Type I Extreme Value distribution. 
$\delta$ is some parameters to be estimated, which can be finite- or infinite-dimensional. 
Consumer $i$ chooses to purchase vehicle with type $j$ if and only if 
\begin{equation*}
U(p_j,x_j,\xi_j,\varepsilon_{ij};\delta)
\geq
U(p_{j'},x_{j'},\xi_{j'},\varepsilon_{ij'};\delta)
 \quad \text{ for }  j'=1,...,J
\end{equation*}

Therefore, the market share of vehicle with type $j$ is given by
\begin{equation*}
  s_j(p,x,\xi;\delta)=\int_{\{\varepsilon: 
  U(p_j,x_j,\xi_j,\varepsilon_{ij};\delta)
\geq
U(p_{j'},x_{j'},\xi_{j'},\varepsilon_{ij'};\delta), \
    j'=1,...,J\}} \text{d} \varepsilon
\end{equation*}

The profit function of a firm $f$, which provides a set of products $J_f$, can be written as
\begin{equation}\label{example1:profit}
  \pi=\sum_{j\in J_f}(p_j-mc_j) M s_j(p,x,\xi;\delta)
\end{equation}
where $M$ represents the size of the market, and $mc_j$ denotes the marginal cost of producing vehicle $j$. 
In notation, write
\begin{equation*}
  \Delta \pi (J_1,J_2, \delta )=\pi (J_1,\delta)  -  \pi (J_2,\delta)
\end{equation*}
as the difference in profit for firm $f$ when offering the product set $J_1$ instead of $J_2$. This captures the net profitability of adding or removing a product type.

Moreover, the sunk cost of supplying vehicle type $j$ is given by $x'_j\eta$, where $\eta$ represents the coefficient for sunk costs. It measures the fixed costs associated with entering or maintaining a product in the market, including investments in production capacity, marketing, and regulatory compliance. Higher values of $x'_j\eta$ indicate that firms face substantial financial barriers when introducing new products. 

Additionally, $\lambda$ is the coefficient that accounts for the salvageable portion of sunk costs. It describes the proportion of the initial investment that can be recovered when a firm discontinues a product type, such as through the resale of equipment or the reallocation of resources. A higher $\lambda$ reduces the effective financial burden of discontinuing a product.



Therefore, the moment inequalities can be constructed based on firm's decision-making process regarding product entry and exit. Let $J_{f,t}$ denote the product set firm $f$ offers at time $t$. 
If at time $t$, firm $f$ introduces product $j$, that is, $j\notin J_{f,t-1}$ and $j\in J_{f,t}$, an entry condition can be written as 
\begin{equation}\label{ex1:moment1}
    \mathbb E_F\left[ \Delta \pi (J_{f,t},J_{f,t}\backslash j, {\delta})-x'_j\eta  \right] \geq  0
\end{equation}
which states the net gain from introducing a product in the market must be at least as large as the sunk costs associated with maintaining it.
Similarly, if at time $t$, firm $f$ no longer provides product $j$, that is, $j\in J_{f,t-1}$ and $j\notin J_{f,t}$, an exit condition can be written as 
\begin{equation}\label{ex1:moment2}
    \mathbb E_F\left[ \Delta \pi (J_{f,t},J_{f,t}\cup j, {\delta})+\lambda x'_j\eta \right] \geq 0
\end{equation}
which imposes that the net loss from obsoleting a product must be compensated by the salvaged sunk costs. These moment inequalities reflect firms' optimal entry and exit strategies in response to profitability conditions and sunk cost constraints.

In this setting, $\lambda$ and $\eta$ are the parameters of interests, and $\delta$ is the nuisance parameter. 
Write $\theta=(\lambda,\eta)$ for convenience. Let $r_1$ be the total number of entry constraint, and $r_2$ be the total number of exit constraint, thus we know $k=r_1+r_2$. Write $J$ as the total number of models.
This moment inequalities fit in the restrictions in Equation \eqref{a}, where 
$B(\theta)$ is a $k \times d_XJ$ matrix, each row of which corresponds to an inequality from Equation \eqref{ex1:moment1} or Equation \eqref{ex1:moment2}, where its $d_X(j-1)+1$ to $d_Xj$ columns are displayed as 
\begin{equation*}
    B(\theta)=
\left(
\begin{array}{ccccccc}
     &  & \vdots  & \vdots & \vdots & &  \\
    0 & \cdots & 0 & \eta' \text{ or } -\lambda\eta' & 0 & \cdots & 0 \\
     &  & \vdots & \vdots & \vdots & &  
\end{array}
\right)
\end{equation*}
"or" in the expression marks if the constraint is from the entry or exit condition. 
$M(W,\theta)$ is a $d_XJ$-vector, written as 
\begin{equation*}
    M(W,\theta) = \left(
    \begin{array}{c}
          x_1 \\
          \vdots \\
          x_J \\
    \end{array}  \right)
\end{equation*}
$C(\theta)$ is a $k\times k$ identity matrix $I_k$, 
and $N(W,\delta)$ is a $k$-vector, where 
\begin{equation*}
    N(W,\delta)=\left[ 
    \begin{array}{c}
         \vdots \\
         \Delta \pi (J_{f,t},J_{f,t}\backslash j, {\delta})  \text{ or }
         \Delta \pi (J_{f,t},J_{f,t}\cup j, {\delta})  \\
         \vdots
    \end{array}
    \right]
\end{equation*}
In this example, $\rho(\theta)$ is $0$.
\end{example}

\begin{example}[Patients Referral]


\cite{ho2014hospital}
considers a model of patients referral. When a patient is referred to some higher-level hospital, the
hospital choice is a result of a complex decision process with the patient and the doctor.
Specifically, it is affected by the price that the insurance company is expected to pay for the patient.
It is also affected by the severity of the patient's condition,
and the distance between patient's and hospital's locations.

They assume this process generates an ordering of the hospitals derived from the patient's and doctor's preferences and assessments.
Let individual $i$ stands for a representative patient. Suppose all insurance companies in the market are indexed as $\Pi=\{ 1,2,...,\bar{\pi} \}$, and $\pi_i \in \Pi$ is the one that patient $i$ is enrolled in.
There is a hospital referral function $W_i(c_i,\pi_i,s_i,h,\theta,\delta)$ as the rating for each hospital $h$ in the insurer's network, whose maximum across $h$ determines the hospital that patient is allocated to. Moreover, it is assumed to take the additively separable form:

\begin{equation*}
  W_i(c_i,\pi_i,s_i,h,\theta,\delta) = \theta_{p,\pi_i} p(c_i,\pi_i,h)+g(q(h),s_i,\delta)+\theta_d d_i(h),
\end{equation*}
where
\begin{itemize}
  \item $p(c_i,\pi_i,h)$ denotes the price the insurance company is expected to pay to hospital $h$ for patient $i$, where $c_i$ stands for the insurance plan that the patient is enrolled in, since one insurance company may offer various plans to costumers.
  \item $s_i$ measures the patient $i$'s severity level.
  \item $d_i(h)$ measures the distance between patient $i$ and hospital $h$.
  \item $\theta_{p,\pi_i}, \ \pi_i=1,2,...,\bar{\pi} $ and $\theta_d$ are the coefficients of the price and the distance respectively, and they are the parameters of interest.
  \item $q(h)$ is a vector of perceived qualities of given hospital $h$, whose dimension is the same as the number of all possible severity levels, and one element for each different severity;
  \item $g(q(h),s_i,\delta)$  measures the impact of the quality of a hospital $h$ for the given severity $s_i$, and $\delta$ is an infinite-dimensional parameters determining the shape of the $g$.
\end{itemize}

Moreover, they assume the validity of revealed preference axiom, which states that the chosen hospital is preferred to all feasible alternatives.
They consider all possible pairs of patients under the same insurance company and severity level but choose different hospitals, while both of whose choices were feasible for both agents.
Within each pair they apply the fact that each patient's choice is preferred to the choice made for the other.

Let the choice of the first patient $i$ be hospital $h$, and the choice of the second patient $i'$ be $h'$.
Define $\Delta W_i(h,h')$ to be the net change of referral function if patient $i$ choose hospital $h$ instead of hospital $h'$, that is
\begin{eqnarray*}
  \Delta W_i(h,h') &=& W_i(c_i,\pi_i,s_i,h,\theta,\delta)- W_i(c_i,\pi_i,s_i,h',\theta,\delta) \\
   &=& \theta_{p,\pi_i} p(c_i,\pi_i,h)    +g(q(h),s_i,\delta)+\theta_d d_i(h)   \\
   & & - \theta_{p,\pi_i} p(c_i,\pi_i,h') -g(q(h'),s_i,\delta)-\theta_d d_i(h'),
\end{eqnarray*}
and the revealed preference theorem implies
\begin{equation}\label{rpt}
   \mathbb E_F\left[ \Delta W_i(h,h') \right] \geq 0.
\end{equation}

To eliminate the nuisance parameter $\delta$, \cite{ho2014hospital} considers the same reasoning of revealed preference theorem for another patient $i'$, which generates
\begin{eqnarray*}
  \Delta W_{i'}(h',h) &=& W_{i'}(c_{i'},\pi_{i'},s_{i'},h',\theta,\delta)- W_{i'}(c_{i'},\pi_{i'},s_{i'},h,\theta,\delta) \\
   &=& \theta_{p,\pi_{i'}} p(c_{i'},\pi_{i'},h')    +g(q(h'),s_{i'},\delta)+\theta_d d_{i'}(h')   \\
   & & - \theta_{p,\pi_{i'}} p(c_{i'},\pi_{i'},h) -g(q(h),s_{i'},\delta)-\theta_d d_{i'}(h)
\end{eqnarray*}
By selecting $i'$ satisfying $\pi_i=\pi_{i'}$ and $s_i=s_{i'}$, 
the summation of these two inequalities cancels out function $g$ as well as the nuisance parameter $\delta$. This process gives
\begin{eqnarray*}
\Delta W_i(h,h')+\Delta W_{i'}(h',h) &=& \theta_{p,\pi_i}
                                    [ p(c_i,\pi_i,h)-p(c_i,\pi_i,h')+p(c_{i'},\pi_{i},h') -p(c_{i'},\pi_{i},h) ]  \\
 & & +\theta_d [ d_i(h)-d_i(h')+d_{i'}(h')-d_{i'}(h) ]  \\
   & \geq & 0
\end{eqnarray*}

However, this procedure has a shortcoming that it only makes use of the paired observations. To constrain the dataset within pairs of patients under the same insurance company and the same severity level will drop a lot of observations, which lead to the loss of available information and may potentially generate a selection bias problem. 
Our proposed method avoids this problem by looking into the moment inequalities system implied by equation (\ref{rpt}) directly.

In this moment inequalities, $\theta_{p,\pi_i}, \ \pi_i=1,2,...,\bar{\pi} $ and $\theta_d$ are parameters of interests, and $\delta$ is the nuisance parameter.  $B(\theta)$ is the vector of $\theta_d$ and $\theta_{p,\pi_i}, \ \pi_i=1,2,...,\bar{\pi} $. $M(W,\theta)$ is composed of $p(c_i,\pi_i,h)$, $p(c_i,\pi_i,h')$, $d_i(h)$, and $d_i(h')$. Moreover,  $C(\theta)$ is just $1$, and $N(W,\theta)$ is composed of $g(q(h),s_i,\delta)$ and $g(q(h'),s_i,\delta)$. 
In this example, $\rho(\theta)$ is $0$.

\end{example}

\subsection{The method of Cox and Shi (2023)}

\cite{cox2023simple} consider a similar structure, where $N(W, \delta)$ does not depend on $W$, i.e., $N(W, \delta) = \delta$. Let $d_{\delta}$ be the dimension of $\delta$. They propose a procedure to eliminate the nuisance $\delta$ by a matrix transformation following Theorem 4.2 of \cite{kohler1967projections}. We briefly review their setup and key steps in their approach. 
For convenience, the notation we use in this section is following their paper and not to be confused with other part of this paper.


Their moment inequality model has a form
\begin{equation*}
   \mathbb E_F\left[ B\bar{M}_n(\theta)-C\delta \right] \leq d
\end{equation*}
where $B$, $C$, and $d$ are $k\times d_M$, $k\times d_{\delta}$, and $k\times 1$ matrices, $\delta$ is an unknown nuisance parameter, $\theta$ is the unknown parameter of interest, and
$\bar M_n(\theta)=\frac{1}{n}M(W_i,\theta)$.
They use matrix reformation to eliminate the nuisance parameter $\delta$. 

\begin{lemma}\label{lemmacs}
    Let $B$ and $C$ be conformable matrices, and $d$ be a conformable vector. There exists a matrix $A(B, C)$ and a vector $b(C, d)$ such that
    \begin{equation*}
        \{\delta : C\delta \geq B\mu - d\} \neq \emptyset \iff A\mu \leq b
    \end{equation*}
Furthermore, $A(B, C) = H(C)B$ and $b(C, d) = H(C)d$, where $H(C)$ is the matrix with rows formed by the vertices of the polyhedron $\{h \in \mathbb{R}^k : h \geq 0, C'h = 0, 1'h = 1\}$.
\end{lemma}

Following Lemma \ref{lemmacs}, the moment inequality system is equivalent to
\begin{equation*}
    A\mathbb E_F[\bar{M}_n(\theta)]\leq b
\end{equation*}
with $A=A(B,C)$ and $b=b(C,d)$. Let $\hat{\Sigma}_n(\theta)$ be an estimator of the covariance matrix of the moments ${\Sigma}_n(\theta)$, where
\begin{equation*}
    {\Sigma}_n(\theta)=\text{Var}(\sqrt{n}\bar{M}_n(\theta))
\end{equation*}
Then they build a test statistics 
\begin{equation*}
T_n(\theta) = \min_{\mu : A\mu \leq b }
n (\bar{M}_n(\theta) - \mu)' \hat{\Sigma}_n(\theta)^{-1}(\bar{M}_n(\theta) - \mu)
\end{equation*}
Let $\hat{\mu}$ denote the solution to the minimization question. 
Let $a_j'$ denote the $j$th row of $A$ and $b_j$ denote the $j$th element of $b$ for $j=1, 2, \dots, d_A$. Let
$$
\hat{J}=\{ j\in\{ 1,2,...,d_A \}: a_j' \hat{\mu}= b_j  \}
$$
which is the indices set for active inequalities.  Write $A_J$ as the submatrix of $A$ formed by the rows of $A$ corresponding to the elements in $J$, and let $\text{rk}(A_J)$ denote the rank of $A_J$. Define $\hat{r} = \text{rk}(A_{\hat{J}})$.

If $\hat{r}=1$, without loss of generality, suppose the first inequality is active and satisfies $a_1\neq 0$. Then for each $j=2,...,d_A$, let
\begin{equation*}
  \hat{\tau}_j= \left\{
  \begin{array}{lcl}
  \frac{ \sqrt{n} \| a_1\|_{\hat{\Sigma}_n} ( b_j-a'_j \hat{\mu} ) }
  { \| a_1\|_{\hat{\Sigma}_n} \| a_j\|_{\hat{\Sigma}_n} -
     a'_1 \hat{\Sigma}_n a_j   }
  & & \text{if }
       \| a_1\|_{\hat{\Sigma}_n} \| a_j\|_{\hat{\Sigma}_n} \neq
                   a'_1 \hat{\Sigma}_n a_j    \\
  \infty & & \text{otherwise }
  \end{array}
   \right.
\end{equation*}
where $\| a \|_{\Sigma} = (a' \Sigma a)^{1/2}$. Let
\begin{equation*}
  \hat{\tau}=\inf_{j\in \{ 2,...,d_A \}} \hat{\tau_j}.
\end{equation*}
Then define
\begin{equation*}
  \hat{\beta}=\left\{ \begin{array}{lcl}
  2\alpha\Phi(\hat{\tau}) & & \text{if } \hat{r}=1  \\
  \alpha & & \text{otherwise}
  \end{array} \right.
\end{equation*}
Then, the RCC test for $H_o:\theta=\theta_0$ is given  by
\begin{equation*}
  \phi_n^{\text{RCC}}(\theta,\alpha)=\mathds{1}\left\{ T_n(\theta)>
  \chi^2_{\hat{r},1-\hat{\beta}} \right\}
\end{equation*}

The following result is a lemma from \cite{cox2023simple}'s argument.

\begin{assumption}\label{assump2}

Given a sequence $\{ (F_n,\theta_n): F_n \in \mathcal{F}, \theta_n \in \Theta_0(F_n) \}_{i=1}^n$,
for every sequence $n_o$, there exists a further subsequence $n_q$ and a sequence of positive definite  matrices $\{ D_q (\theta_{n_q}) \}$ such that
\begin{equation*}
  \sqrt{n_q}D_q^{-1/2}( \bar{p}_{n_q}(\theta_{n_q})-E[\bar{p}_{n_q}(\theta_{n_q})] ) \xrightarrow{d} N(0,\Omega)
\end{equation*}
for some positive definite correlation matrix $\Omega$, and
\begin{equation*}
  \| D_q(\theta_{n_q})^{-1/2} \hat{\Sigma}_{n_q}(\theta_{n_q}) D_q(\theta_{n_q})^{-1/2} - \Omega  \| \xrightarrow{p} 0
\end{equation*}
\end{assumption}

This assumption is actually weak, as \cite{cox2023simple} argues, since it can be implied by four ordinary conditions: i.i.d. dataset, positive variance of moments, significant covariance matrix determinant, and bounded standardized moments. It elicits the following result:

\begin{lemma}
Suppose the Assumption \ref{assump2} holds for all sequences $\{ (F_n,\theta_n): F_n \in \mathcal{F}, \theta_n \in \Theta_0(F_n) \}_{i=1}^n$,
then
\begin{equation*}
  \limsup_{n\rightarrow \infty} \sup_{F\in \mathcal{F}} \sup_{\theta \in \Theta_0(F)} 
  E[ \phi_n^{RCC}(\theta,\alpha)  ] \leq \alpha
\end{equation*}
\end{lemma}
This lemma has illustrated the validity of \cite{cox2023simple}'s method. 
However, their method cannot be directly applied to our structure since when $N(W, \delta)$ is no longer a single parameter.


\subsection{Proposed inference methods}\label{sec:method}
To address this issue, we propose a two-step estimation procedure. The first step is an estimation of the nuisance parameter, where the nuisance parameter is assumed to be point identified. Suppose the first stage estimation of the nuisance parameter $\delta$ is $\hat{\delta}$. We will make further assumptions on the properties of $\hat{\delta}$ later to ensure the validity of our method. 
The method in the first stage estimation can vary across realistic conditions. 

In the second step, we can follow the approach of \cite{cox2023simple} but make some adjustments. First, denote the sample average of $X$ as $\bar{X}_n = n^{-1}\sum_{i=1}^n X_i$, 
where $X$ can be any function of random variables. Write 
\[
p(W, \theta, \delta) = \begin{bmatrix} M(W, \theta) \\ N(W,\theta,\delta) \end{bmatrix}
\]
Denote the sample mean by $\bar M(\theta) =  \frac{1}{n}\sum_{i=1}^n M(W_i,\theta)$, $\bar N(\theta,\hat\delta) =  \frac{1}{n}\sum_{i=1}^n N(W_i,\theta,\hat\delta)$, and $\bar{p}_n(\theta,\hat\delta) = \frac{1}{n}\sum_{i=1}^n p(W_i,\theta,\hat\delta)$. 
Let $\hat{\Sigma}_n(\theta, \hat{\delta})$ be an invertible estimator of the conditional variance $\text{Var}(\sqrt{n}\bar{p}_n(\theta, \hat{\delta}) | \hat{\delta})$. 
Construct the following test statistic:
\begin{equation*}
T_n(\theta, \hat{\delta}) = \min_{t, \eta : B(\theta)t - C(\theta)\eta \leq \rho(\theta)}
n (\bar{M}(\theta) - t, \bar{N}(\theta, \hat{\delta}) - \eta)' \hat{\Sigma}_n(\theta, \hat{\delta})^{-1}(\bar{M}(\theta) - t, \bar{N}(\theta, \hat{\delta}) - \eta).
\end{equation*}

Write 
$\kappa =[t,\eta]'$, 
then the constraint $B(\theta)t - C(\theta)\eta \leq \rho(\theta)$ can be written as
\begin{equation*}
[B(\theta), -C(\theta)] \kappa \leq \rho(\theta).
\end{equation*}
Let $A(\theta) = [B(\theta), -C(\theta)]$, the test statistic then becomes
\begin{equation}\label{teststat}
T_n(\theta, \hat{\delta}) = \min_{\kappa : A(\theta) \kappa \leq \rho(\theta)}
n (\bar{p}(\theta, \hat{\delta}) - \kappa)' \hat{\Sigma}_n(\theta, \hat{\delta})^{-1}(\bar{p}(\theta, \hat{\delta}) - \kappa).
\end{equation}

Let $\hat{\kappa}(\theta, \hat{\delta})$ be the solution to the minimization problem mentioned previously. Let $a_l'(\theta)$ denote the $l$th row of $A(\theta)$ and $d_l(\theta)$ denote the $l$th element of $\rho(\theta)$ for $l=1, 2, \dots, d_A$, where $d_A$ is the dimension of rows of $A(\theta)$. Define 
$$
\hat{J}(\theta,\hat{\delta})=\{ l\in\{ 1,2,...,d_A \}: a_l'(\theta) \hat{\kappa}(\theta,\hat{\delta})= d_l(\theta)  \}
$$
to collect the indices associated with the restrictions that hold with equality. Write $A_J$ as the submatrix of $A$ formed by the rows of $A$ corresponding to the elements in $J$, and let $\text{rk}(A_J)$ denote the rank of $A_J$. Define $\hat{r}(\theta, \hat{\delta}) = \text{rk}(A_{\hat{J}(\theta, \hat{\delta})}(\theta))$.

Conditional on the set of active inequalities, we follow the refined chi-squared  test as proposed by \cite{cox2023simple}. When $\hat{r} \neq 1$, the critical value of the RCC test statistic is set to be the $\chi^2$-distribution with a degree of freedom $\hat{r}$ at the quantile $1-\alpha$, that is, $\chi^2_{\hat{r}, 1-\alpha}$, where $\alpha$ is the pre-supposed significant level. When $\hat{r} = 1$, this critical value can be too conservative, as the test statistic tends to mass at zero. To address this issue, the RCC test employs a quantile of $1-\beta_n(\theta, \hat{\delta})$ instead of $1-\alpha$, which depends on how far from being active these inequalities are and would restore the size of the test.

Without loss of generality, suppose the first inequality is active. To measure the distance from being active for other inequalities, define $z_j(\theta, \hat{\delta})$, $j=2, 3, \dots, d_M+d_N$, as
\begin{equation*}
  z_j(\theta,\hat{\delta})= \left\{
  \begin{array}{lcl}
  \frac{ \sqrt{n} \| a_1\|_{\hat{\Sigma}_n} ( d_j-a'_j \hat{\kappa} ) }
  { \| a_1\|_{\hat{\Sigma}_n} \| a_j\|_{\hat{\Sigma}_n} -
     a'_1 \hat{\Sigma}_n a_j   }
  & & \text{if }
       \| a_1\|_{\hat{\Sigma}_n} \| a_j\|_{\hat{\Sigma}_n} \neq
                   a'_1 \hat{\Sigma}_n a_j    \\
  \infty & & \text{otherwise }
  \end{array}
   \right.
\end{equation*}
where $\| a \|_{\Sigma} = (a' \Sigma a)^{1/2}$ in definition. $z_j$ is equal to zero when the $j$th inequality is active, and positive when it is inactive. Then, let
\begin{equation*}
  z(\theta,\hat{\delta})=\inf_{j\in \{ 2,...,d_M+d_N \}} z_j(\theta,\hat{\delta}).
\end{equation*}

Define
\begin{equation*}
  \beta_n(\theta,\hat{\delta})=\left\{ \begin{array}{lcl}
  2\alpha\Phi(z(\theta,\hat{\delta})) & & \text{if } \hat{r}(\theta,\hat{\delta})=1  \\
  \alpha & & \text{otherwise}
  \end{array} \right.
\end{equation*}
where $\Phi$ is the standard normal cumulative distribution function. 
Finally, the critical value of the RCC test is the $\chi^2$-distribution with a degree of freedom $\hat{r}$ at the quantile $1-\beta_n(\theta, \hat{\delta})$. The test for $H_0: \theta_0 = \theta$ is represented as
\begin{equation*}
  \phi_n^{\text{RCC}}(\theta,\hat{\delta},\alpha)=\mathds{1}\left\{ T_n(\theta,\hat{\delta})>
  \chi^2_{\hat{r},1-\beta_n(\theta,\hat{\delta})} \right\}
\end{equation*}
Therefore, a confidence interval for $\theta$ under the significant level $\alpha$ can be built as
\begin{align}
	\{ \theta \in \Theta: \phi_n^{RCC}(\theta,\hat{\delta},\alpha)=0 \}.
\end{align}
after plugging in the first stage estimator $\hat{\delta}$.


\section{Asymptotics and Finite-Sample Validity}\label{ch:asymptotics}

In large samples, to ensure the size of the test, we have to make some high-level assumptions on the first-stage estimator $\hat{\delta}$.
Since we assume point identification of $\delta$, for a given data generating process $F\in \mathcal{F}$ there is only one $\delta$ satisfying $\delta=\Delta_0(F)$. Write the first-stage estimator as $\hat{\delta} = \hat{\Delta}_0(F)$,
we post another assumption to motivate the theorem of the asymptotic properties.

\begin{assumption}\label{assump3}
Given a sequence $\{ (F_n,\theta_n,\delta_n,\hat{\delta}_n): F_n \in \mathcal{F}, \theta_n \in \Theta_0(F_n), \delta_n = \Delta_0(F_n), \hat{\delta}_n = \hat{\Delta}_0(F_n)  \}_{i=1}^n$,
for every subsequence $n_o$, there exists a further subsequence $n_q$ and a sequence of positive definite $(d_M+d_N)^2$ matrices $\{ D_q (\theta_{n_q},\delta_{n_q}) \}$ such that
\begin{equation*}
  \sqrt{n_q}D_q(\theta_{n_q},\delta_{n_q})^{-1/2} ( \bar{p}_{n_q}(\theta_{n_q},\hat{\delta}_{n_q})-E[\bar{p}_{n_q}(\theta_{n_q},\delta_{n_q})] ) \xrightarrow{d} N(0,\Omega)
\end{equation*}
for some positive definite correlation matrix $\Omega$, and
\begin{equation*}
  \| D_q(\theta_{n_q},\delta_{n_q})^{-1/2} \hat{\Sigma}_{n_q}(\theta_{n_q},\hat{\delta}_{n_q}) D_q(\theta_{n_q},\delta_{n_q})^{-1/2} - \Omega  \| \xrightarrow{p} 0
\end{equation*}
\end{assumption}

This assumption is a variant of Assumption \ref{assump2}. In large samples, we assume the estimator $\hat{\delta}$ can ensure the asymptotic normality of $\bar{p}$ and the convergence of $\hat{\Sigma}$.
If both $\theta$ and $\delta$ considered as parameters of interest and the first stage estimation is skipped, 
the assumption is directly implied by Assumption \ref{assump2}. Nevertheless, this assumption is more than the consistency of the first-stage estimator $\hat{\delta}$, as it involves the large sample performance of $\theta_n$, $\delta_n$, and $\hat{\delta}_n$ simultaneously. 
%
The first-stage estimator $\hat\delta$ could have an impact on computing the variance  $\hat{\Sigma}_n(\theta,\hat\delta)$. As exemplified later in Section \ref{sec:empirical}, we use an influence function to propagate the error into the second stage. We can only disregard this impact in the  case where $E[p(x,\theta,\delta)]$ is orthogonal to the first-stage estimator $\hat\delta$, which means the Jacobian matrix derived later in Equation \eqref{eq:jacob} is zero evaluated at $\delta=\hat\delta$, and thus the influence function does not contribute to the sample variance.

We have the following theorem illustrating that the test has correct asymptotic size in large samples.

\begin{theorem}\label{th2}
Suppose the assumption \ref{assump3} holds for all sequences $\{ (F_n,\theta_n,\delta_n,\hat{\delta}_n): F_n \in \mathcal{F}, \theta_n \in \Theta_0(F_n), \delta_n = \Delta_0(F_n), \hat{\delta}_n = \hat{\Delta}_0(F_n) \}_{i=1}^n $,
then
\begin{equation*}
  \limsup_{n\rightarrow \infty} \sup_{F\in \mathcal{F}} \sup_{\theta \in \Theta_0(F)}
  E_F[ \phi_n^{RCC}(\theta,\hat{\delta},\alpha)  ] \leq \alpha
\end{equation*}
\end{theorem}

Therefore, when the number of the observations $n$ approximates infinity, the confidence set constructed from test inversion will asymptotically cover the parameters of interest at the presumed significant level.

In finite samples, this method is valid under certain conditions. The following theorem states that
when sample mean is assumed to be normal and the estimation of covariance matrix is almost surely accurate,
the test exhibits proper size under finite sample.

\begin{theorem}\label{th1}
Under the following assumptions:
\begin{enumerate}
  \item Conditional on $\hat{\delta}$, there is $\sqrt{n}( \bar{p}(\theta,\hat{\delta})- E[\bar{p}(\theta,\hat{\delta})| \hat{\delta} ] )  \sim N(0,\Sigma_n(\theta,\hat{\delta}))$, where $\Sigma_n(\theta,\hat{\delta})$ is a positive definite and invertible matrix for all $\theta \in \Theta$ and $\hat{\delta} \in \Delta$.
  \item $\hat{\Sigma}_n(\theta,\hat{\delta})=\Sigma_n(\theta,\hat{\delta})$ a.s. for all $\theta \in \Theta$ and $\hat{\delta} \in \Delta$.
\end{enumerate}
Then, for any $\theta \in \Theta$ and $\hat{\delta} \in \Delta$, $E[\phi_n^{RCC}(\theta,\hat{\delta},\alpha) | \hat{\delta} ] \leq \alpha$. Equality holds when $A(\theta) E[\bar{p}_n(\theta,\hat{\delta}) | \hat{\delta} ]=\rho(\theta)$ and $A(\theta)\neq \textbf{0}$.
\end{theorem}

Theorem \ref{th1} demonstrates the validity of the proposed test under finite sample conditions. Specifically, it ensures that the test maintains the correct size, assuming normality of the sample mean and accurate estimation of the covariance matrix.

\section{Empirical Application}\label{sec:empirical}

This section continues Example \ref{example1} by applying the methodology to the empirical context of the U.S. commercial vehicle market using the dataset from \cite{wollmann2018trucks}. 
The objective is to provide empirical evidence on the relative magnitude of entry and exit costs across different vehicle attributes in the U.S. commercial truck market.


The dataset records transaction data and product characteristics for each brand and model within a year. The data spans from 1992 to 2012, covering sales information for every type of vehicle product within all mainstream brands.

The characteristics in the dataset include Gross Vehicle Weight Rating, Cab-Over-Engine Indicator, Compact-Front-End Indicator, and Long-Cap Indicator. Gross Vehicle Weight Rating measures the legal maximum load a vehicle can carry and is a key specification used by both regulators and manufacturers to classify vehicles. It also reflects the intended usage of the vehicle, with higher ratings typically indicating commercial applications. 
The Cab-Over-Engine Indicator captures whether the vehicle has a cab-over-engine design, where the driver sits directly above the engine and front axle, with improving maneuverability and visibility but reducing comfort.
The Compact-Front-End Indicator identifies vehicles with a shortened front hood that balances the benefits of maneuverability and space while accommodating limited engine capacity. Lastly, the Long-Cap Indicator denotes whether a vehicle is equipped with an extended cab, which offers additional interior space and comfort, especially suited for long-distance or demanding operations.

In our analysis, we treat each brand as a firm and define two vehicles as being within the same type if their characteristics are identical. Consequently, all vehicles are classified into 30 distinct types.

In the first stage of our estimation, we recover the unobserved mean utilities and estimate the demand-side parameters using a procedure inspired by \cite{berry1995automobile} and \cite{nevo2001measuring}. 
We begin with the utility specification for consumer $i$ choosing product $j$:
\begin{equation} \label{eq:utility}
  U(p_j,x_j,\xi_j,\varepsilon_{ij};\delta)
     = x_j \beta - \alpha p_j + \xi_j + \varepsilon_{ij},
\end{equation}
where $x_j$, $p_j$, $\xi_j$, and $\varepsilon_{ij}$ are defined in Example \ref{example1}. $\delta=(\beta,\alpha)$  are random coefficients that capture consumer heterogeneity, and are assumed to follow a joint distribution $F_{\beta,\alpha}(\cdot)$. 
We denote it by $U_{ij}$ for simplicity.

Given the assumption on $\varepsilon_{ij}$, the predicted market share of product $j$ takes the multinomial logit form:
\begin{equation} \label{eq:market_share}
    s_j = \int \frac{\exp\left(x_j \beta - \alpha p_j + \xi_j\right)}{1+\sum_{k=1}^J \exp\left(x_k \beta - \alpha p_k + \xi_k\right)} dF_{\beta,\alpha},
\end{equation}
where the integration is with respect to the joint distribution $F_{\beta,\alpha}$. Because the integral in Equation~\eqref{eq:market_share} does not yield a closed-form expression, we approximate it using Monte Carlo simulation. That is, for $R$ independent draws $(\beta_r,\alpha_r)$, the simulated market share is computed as
\begin{equation}\label{eq:simulationdraw}
    s_j \approx \frac{1}{R} \sum_{r=1}^{R} \frac{\exp\left(x_j \beta_r - \alpha_r p_j + \xi_j\right)}{1+\sum_{k=1}^J \exp\left(x_k \beta_r - \alpha_r p_k + \xi_k\right)}.
\end{equation}

Write the mean utility for product $j$ as
\begin{equation}\label{eq:zeta}
    \zeta_j \equiv x_j \beta - \alpha p_j + \xi_j
\end{equation}
To match the predicted market shares with the observed ones, we recover the vector $\zeta = (\zeta_1, \ldots, \zeta_J)$ using a fixed-point iteration. Let $\hat{s}_j$ denote the observed market share for product $j$. We start with an initial guess 
$
\zeta_j^{(0)} = \ln(\hat{s}_j)
$, 
and update iteratively as follows:
\begin{equation}
    \zeta_j^{(t+1)} = \zeta_j^{(t)} + \ln(\hat{s}_j) - \ln\big(s_j( \zeta^{(t)})\big),
\end{equation}
until the change in $\zeta$ is below a specified tolerance. 

A major econometric challenge in estimating demand is the potential endogeneity of prices $p_j$, which may be correlated with unobserved product quality $\xi_j$. To address this, we use an instrumental variable, $z_j$, which is assumed to satisfy
\[
E[z_j \xi_j] = 0.
\]
Based on our model, the GMM objective function is then
\begin{equation} \label{eq:gmm}
    Q(\delta) = \xi' Z W Z' \xi,
\end{equation}
where $Z$ is the instrument matrix, 
$\xi$ is the vector of $\xi_j$ after plugging the estimated mean utility $\zeta_j$ to Equation \eqref{eq:zeta}, 
and $W$ is a chosen positive-definite weighting matrix. 
\footnote{In our application, the parameter vector $\theta$ is estimated by minimizing $Q(\delta)$ using the BFGS algorithm.}

Note that the simulation draws in Equation \eqref{eq:simulationdraw} are not assumed to represent the true distribution of the random coefficients; rather, they are used as a computational device to approximate the integral in the market share model, which provides sufficient variation for the simulation while not overwhelming the signal in the data.

After the simulation step, the fixed-point iteration recovers the unobserved mean utilities $\zeta$ based on these simulated market shares. Although the initial draws for $\beta$ and $\alpha$ may be crude, they serve only as starting values in the simulation of market shares. In the subsequent GMM estimation step, these parameters are updated by minimizing the GMM objective function based on Equation \eqref{eq:gmm}. 
Thus, even if the initial simulation of $\beta$ and $\alpha$ uses a simple normal distribution, the iterative GMM procedure refines these parameters to achieve consistency with the observed market shares and the imposed moment conditions.


Table \ref{tab:results} presents the estimation results of the coefficients of product characteristics $\beta$ and the price sensitivity $\alpha$. 

\begin{table}[htbp]
    \centering
    \begin{tabular}{l r}
        \hline
        Parameter & Estimate \\
        \hline
        $\beta_1$ (Gross Vehicle Weight Rating) & $2.61\times10^6$ \\
        $\beta_2$ (Cab-Over-Engine Indicator)     & $7.70\times10^5$ \\
        $\beta_3$ (Compact-Front-End Indicator)        & $1.58\times10^5$ \\
        $\beta_4$ (Long-Cap Indicator)                & $6.26\times10^6$ \\
        \hline
        $\alpha$ (Price Sensitivity)              & $1.33\times10^5$ \\
        \hline
    \end{tabular}
    \caption{First-stage parameter estimates for the U.S. vehicle market}
    \label{tab:results}
\end{table}

The estimated coefficients in Table~\ref{tab:results} are large in magnitude, which is driven by the scaling of the explanatory variables and the nature of multinomial logit model. The signs of the parameters are consistent with economic theory. The positive estimates for the $\beta$ coefficients indicate that enhancements in product characteristics increase the mean utility, whereas the positive value of $\alpha$ implies that higher prices detract from consumer utility. Among all the features of vehicles, an extended cap is most valued by drivers. A cab over the engine is also cherished, while a compact front end type is less valued by drivers. 
This first-stage estimation thus produces demand-side parameters that are later used to construct the structural model for product entry and exit decisions in the second stage.

Recall that $p(W, \theta, \delta)$ is constructed as
$p(W,\theta, \delta) = [M(W,\theta)', N(W,\theta, \delta)']'$. 
In our framework, the  variance of the sample average $\bar{p}_n(\theta,\hat\delta)$ 
must account not only for the intrinsic sampling variation in $p(W,\theta,\delta)$, but also for the uncertainty arising from the estimation error in the nuisance parameter $\hat\delta$ from the first stage. 
To estimate the variance of $\bar{p}_n(\theta,\hat\delta)$, we apply the influence function method introduced by \cite{newey1994large} to capture how errors in the estimation of $\delta$ propagate into the second-stage estimation.

Let $P_{\delta}$ denote the Jacobian matrix:
\begin{equation*}
    P_{\delta} = \frac{\partial \mathbb{E}[p(W,\theta,\delta)]}{\partial \delta},
\end{equation*}
and denote the influence function by $\psi(W,\delta)$ to correct for the error in estimating $\delta$, defined as
\begin{equation*}
    \psi(W,\delta) = -G^{-1} g(W,\delta) 
\end{equation*}
with
\begin{eqnarray*}
    g(W, \delta) &=& z_j \cdot (\delta_j - x_j \beta + \alpha p_j) \\
    G &=& \mathbb{E}\left[\frac{\partial g(W,\delta)}{\partial \delta}\right]
\end{eqnarray*}

The corrected moment for each observation is then given by 
$p(W,\theta,\delta) + P_{\delta}\psi(W,\delta)$, 
and consequently, the variance of $\sqrt{n}\bar{p}_n(\theta,\hat\delta)$ can be estimated as
\begin{equation} \label{eq:variance_estimator}
\hat \Sigma_n(\theta,\hat{\delta}) = \frac{1}{n} \sum_{i=1}^n \left[p(W_i,\theta,\hat{\delta}) + P_{\hat\delta}\psi(W_i,\hat{\delta})\right] \left[p(W_i,\theta,\hat{\delta}) + P_{\hat\delta}\psi(W_i,\hat{\delta})\right]'.
\end{equation}
where
\begin{equation}\label{eq:jacob}
  P_{\hat\delta} = \frac{1}{n} \sum_{i=1}^n \left. \frac{\partial p(W_i, \theta, \delta)}{\partial \delta} \right|_{\delta=\hat\delta}
\end{equation}


The influence function correction compensates for the additional variability due to the estimation error in $\delta$, as well as for any potential endogeneity in prices via the instrumental variable $z_j$. 
The resulting variance estimator in Equation~\eqref{eq:variance_estimator} is then used to conduct robust inference on the structural parameters in our model.




Having obtained the first-stage estimates $\hat{\delta}$ and the sample variance of the moment function $\hat{\Sigma}_n(\theta,\hat{\delta})$, we proceed to the second-stage estimation of the sunk-cost parameters $\theta = (\lambda, \eta)$ using the RCC test procedure described in Section~\ref{sec:method}. Since the model imposes moment inequality restrictions, the parameters are only partially identified. As a result, unlike point-identified settings where a single estimate summarizes the inference, our procedure yields a high-dimensional confidence region that characterizes the identified set of $\theta$.

In practice, we evaluate the test statistic $T_n(\theta)$ over a finite grid in the $(\lambda,\eta)$ space and retain those values of $\theta$ for which the RCC test does not reject at level $\alpha$, that is, for which $\phi_n^{RCC}(\theta,\hat{\delta},\alpha) = 0$. These values constitute our estimate of the identified set.

To visualize this high-dimensional confidence region, we plot two-dimensional slices in the $(\lambda, \eta_1)$ plane while holding the remaining components $(\eta_2, \eta_3, \eta_4)$ fixed at representative values. Specifically, Figure~\ref{fig:conf_regions} presents four panels, each corresponding to a different choice of $(\eta_2, \eta_3, \eta_4)$. The shaded areas in these plots represent combinations of $(\lambda, \eta_1)$ that are included in the confidence region under the specified values of the other components.

\begin{figure}[htbp]
  \centering
  \begin{subfigure}[b]{0.48\textwidth}
    \includegraphics[width=\textwidth]{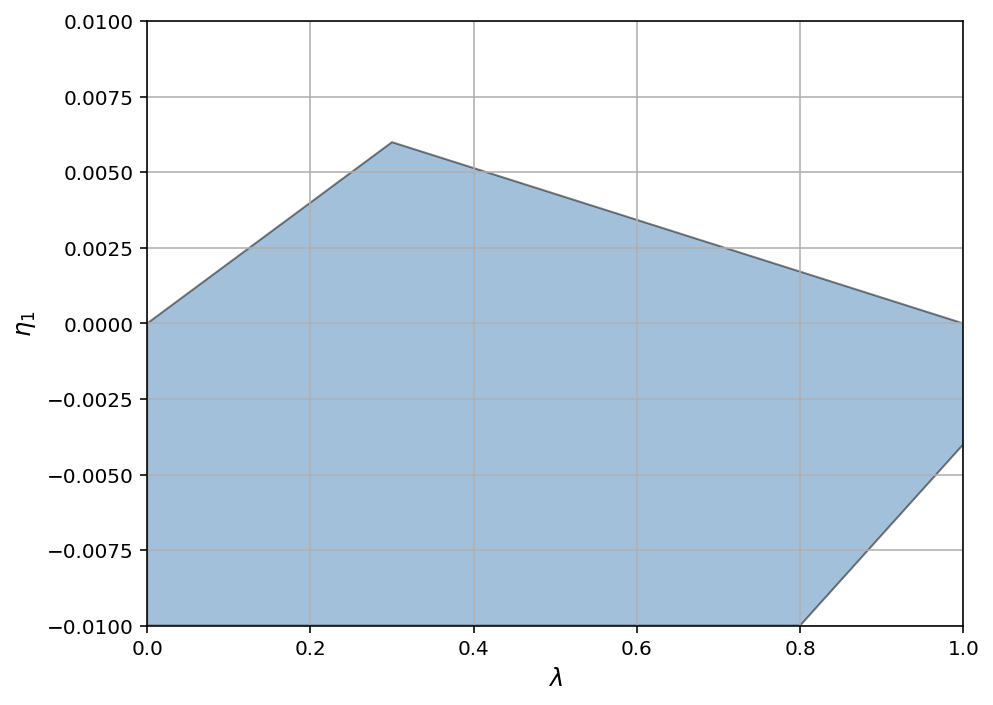}
    \caption{$(\eta_2,\eta_3,\eta_4)=(0.10,\,0.10,\,0.05)$}
    \label{fig:conf_a}
  \end{subfigure}
  \begin{subfigure}[b]{0.48\textwidth}
    \includegraphics[width=\textwidth]{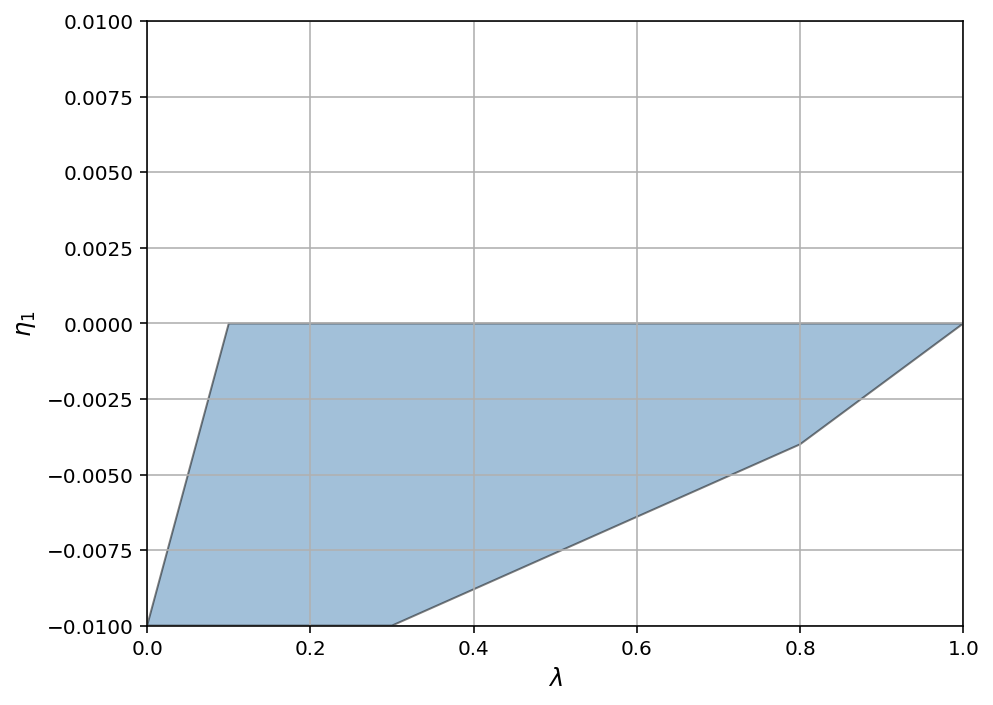}
    \caption{$(\eta_2,\eta_3,\eta_4)=(0.20,\,0.10,\,0.05)$}
    \label{fig:conf_b}
  \end{subfigure}
  
  \begin{subfigure}[b]{0.48\textwidth}
    \includegraphics[width=\textwidth]{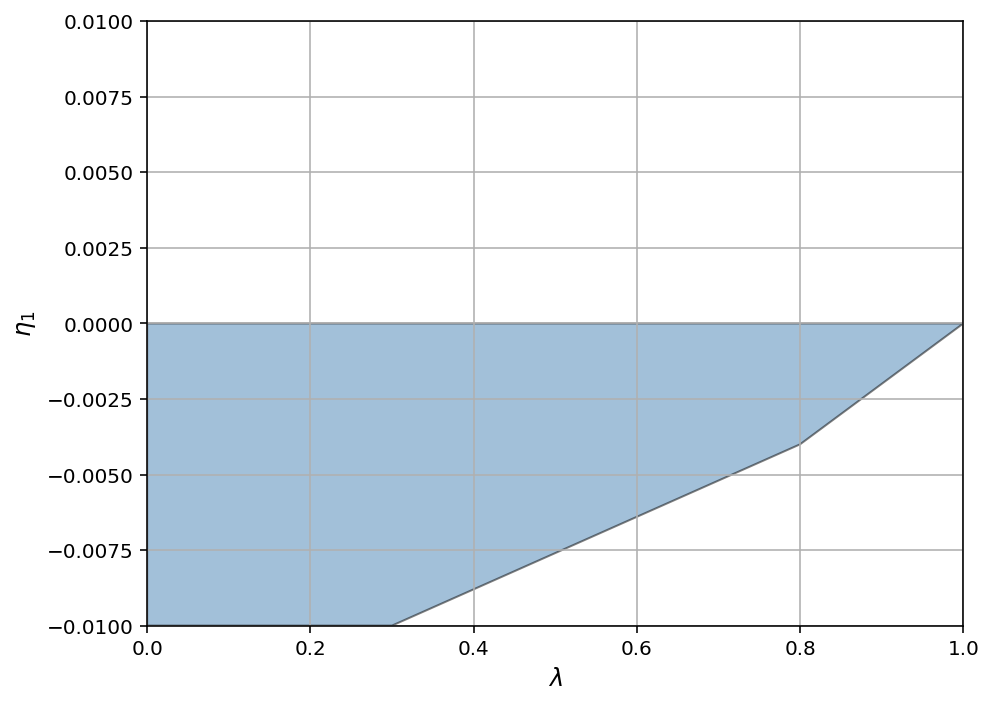}
    \caption{$(\eta_2,\eta_3,\eta_4)=(0.10,\,0.05,\,0.05)$}
    \label{fig:conf_c}
  \end{subfigure}
  \begin{subfigure}[b]{0.48\textwidth}
    \includegraphics[width=\textwidth]{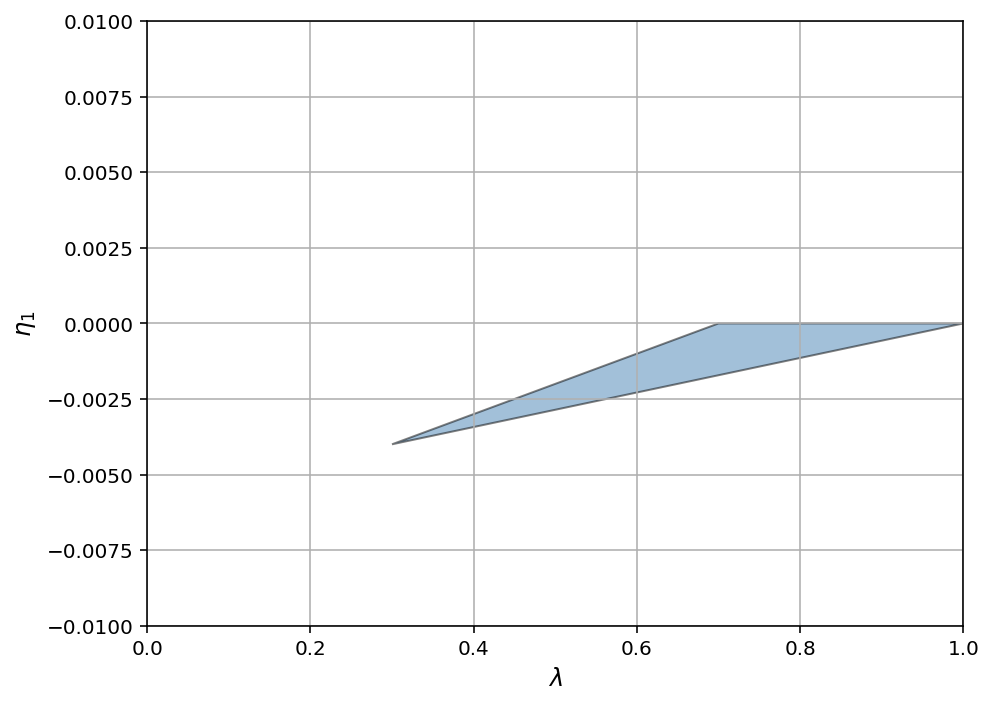}
    \caption{$(\eta_2,\eta_3,\eta_4)=(0.10,\,0.10,\,0.10)$}
    \label{fig:conf_d}
  \end{subfigure}
  \caption{Estimated 95\% confidence regions for $(\lambda,\eta_1)$ at fixed slices of $(\eta_2,\eta_3,\eta_4)$.  Each point in a panel represents a grid value of $(\lambda,\eta_1)$ that satisfies the RCC test.}
  \label{fig:conf_regions}
\end{figure}

We observe that the confidence region is most concentrated around $\lambda \approx 0.3$ and $\eta_1 \approx 0$, suggesting that the sunk costs exhibit low recoverability, and the gross vehicle weight does not have a significant influence on the sunk cost of establishing a production line. 
Moreover, 
changes in the remaining components $(\eta_2, \eta_3, \eta_4)$ affect the shape of the confidence region in the $(\lambda, \eta_1)$ plane. 
Starting from the baseline configuration $(\eta_2, \eta_3, \eta_4) = (0.10,\,0.10,\,0.05)$, increasing $\eta_2$, the sunk cost sensitivity to the cab-over-engine feature, 
or decreasing $\eta_3$, the sunk cost sensitivity to compact-front-end design, 
will slightly narrow the confidence region. This suggests that $(\eta_2, \eta_3, \eta_4) = (0.20,\,0.10,\,0.05)$ or $(0.10,\,0.05,\,0.05)$ is still consistent with the data and lie within the identified set. 
In contrast, increasing $\eta_4$, which captures the effect of long-cap design on sunk cost, 
results in a substantially reduced confidence region. These patterns indicate that the model imposes tighter restrictions on the influence of long-cap design characteristics on sunk costs, thus narrowing the plausible values of $(\lambda, \eta_1)$ that remain consistent with observed entry and exit decisions.

\section{Conclusion}\label{ch:conclusion}

In this paper, we have proposed an econometric method for estimating partially identified parameters in moment inequalities and separable nuisance parameters. Our method is applicable to a wide range of models, and we have demonstrated its effectiveness by applying it to two examples: the US vehicle market model from \cite{wollmann2018trucks} and the hospital referral model from \cite{ho2014hospital}. The former example focuses on the structural estimation of vehicle markets, while the latter examines the complex decision process of patient referrals to hospitals.

Our  method ensures correct asymptotic size in large samples. To achieve this, we have made high-level assumptions on the first-stage estimator $\hat{\delta}$. Under these assumptions, the confidence set constructed from test inversion will asymptotically cover the parameters of interest at the presumed significant level when the number of observations $n$ approaches infinity.

In conclusion, the proposed method offers a robust and versatile approach for estimating and testing models with moment inequalities and separable nuisance parameters. It is applicable to various types of models, as demonstrated by our examples, and exhibits valid finite-sample and asymptotic properties. Our findings contribute to the econometric literature and provide researchers with a powerful tool for analyzing complex economic models.

Future research directions could include exploring methods to relax the separability assumption, which may extend the applicability of the proposed estimation method to a broader range of models. Additionally, investigating alternative assumptions on the first-stage estimator could provide insights into the robustness and performance of the method under different estimation scenarios. This would further enhance the understanding of the practical implications of the method in various econometric applications.

\bibliographystyle{ecta}
\bibliography{citation.bib}

\appendix

\section{Proofs}

\begin{proof}[Proof of Theorem \ref{th1}]

In this proof, firstly assume $\Sigma_n(\theta,\hat{\delta})=nI_{d_M+d_N}$. This assumption does not lose any generality since if it's not the case, we can restore the diagonal variance case by pre-multiply $\bar{p}(\theta,\hat{\delta})$ by $n^{1/2}\Sigma_n(\theta,\hat{\delta})^{-1/2}$, and post-multiply $A$ by $n^{-1/2}\Sigma_n(\theta,\hat{\delta})^{1/2}$.

Fix $\theta$ and condition on $\hat{\delta}$, write $X=\bar{p}(\theta,\hat{\delta})\sim N(\mu,I_{d_M+d_N})$, where $\mu=E[\bar{p}(\theta,\hat{\delta})]$. Define a set $F=\{ \mu_0 \in R^{d_M+d_N} | A\mu_0 \leq \rho \}$, and note that the moment inequalities imply $\mu\in F$. After this simplification, the statistics in (\ref{teststat}) can be written as
\begin{equation*}
  T_n(\theta,\hat{\delta})= \| X-\hat{\mu} \|^2,
\end{equation*}
where $\hat{\mu}$ is the projection of $X$ onto $F$, which we denote by $P_F X$.
For convenience, we temporarily drop the argument of $(\theta,\hat{\delta})$.
The solution $\hat{\kappa}=[\hat{\kappa}_1,...,\hat{\kappa}_{d_M+d_N}]'$ to the minimization problem in (\ref{teststat}) is hence
\begin{equation*}
  \hat{\kappa}_j=\left\{
  \begin{array}{lcl}
    \frac{\| a_1 \| (d_j-a'_j\hat{\mu})}{\| a_1 \| \| a_j \| - a'_1 a_j} &  & \text{if } \| a_1 \| \| a_j \| \neq a'_1 a_j \\
    \infty &  & \text{otherwise}
  \end{array}
  \right.
\end{equation*}

Therefore, we define $K_J$, $J\in \{ 1,...,d_A \}$ within the same way as \cite{cox2023simple}, which forms a partition of $\mathbb{R}^{d_M+d_N}$, then we obtain
\begin{eqnarray*}
   & & \text{Pr}(\| X-P_CX \|^2> \chi^2_{\hat{r},1-\beta_n}) \\
   &=& \sum_{J \subseteq \{ 1,...,d_A \} } \text{Pr}( X\in K_J \text{ and } \| X-P_CX \|^2> \chi^2_{\hat{r},1-\beta_n} )
\end{eqnarray*}
which follows Lemma 3 by \cite{cox2023simple}. The rest of the proof are therefore the same as \cite{cox2023simple}'s appendix A. Therefore, the conclusion of the theorem follows.
\end{proof}

\begin{proof}[Proof of Theorem \ref{th2}]
Let $\{ F_n, \theta_n, \delta_n \}$ be an arbitrary sequence satisfying $F_n\in \mathcal{F}$, $\theta_n \in \Theta_0(F_n)$, and $\delta_n =\Delta_0(F)$ for all $n$.
Since $0 \leq \phi_n^{RCC}(\theta,\hat{\delta},\alpha) \leq 1$, and $[0,1]$ is naturally compact,
we can find a subsequence $\{ n_m \}$ of $\{n\}$ such that
$E_F[ \phi_{n_m}^{RCC}(\theta_{n_m},\hat{\delta}_{n_m},\alpha)  ]$ is convergent and satisfies
$$
\limsup_{n\rightarrow \infty} \sup_{F\in \mathcal{F}} \sup_{\theta \in \Theta_0(F)}
  E_F[ \phi_n^{RCC}(\theta,\hat{\delta},\alpha)  ]
  =
  \limsup_{m\rightarrow \infty}
  E_F[ \phi_{n_m}^{RCC}(\theta_{n_m},\hat{\delta}_{n_m},\alpha)  ]
$$

Therefore, a sufficient condition to prove the theorem is to establish that:
there exists a further subsequence $\{ n_q \}$ of $\{ n_m \}$ such that,
\begin{equation*}
  \lim_{q\rightarrow \infty} \text{Pr}_{F_{n_q}} (T_{n_q}(\theta_{n_q},\hat{\delta}_{n_q}) \leq \chi^2_{\hat{r},1-\hat{\beta}} ) \geq 1-\alpha
\end{equation*}
since the definition implies
$$
 E_F[ \phi_{n_q}^{RCC}(\theta_{n_q},\hat{\delta}_{n_q},\alpha)  ] =
\text{Pr}_{F_{n_q}} (T_{n_q}(\theta_{n_q},\hat{\delta}_{n_q}) \leq \chi^2_{\hat{r},1-\hat{\beta}} )
$$

By the assumption, for the arbitrary sequence $\{n_m\}$, there exists a further subsequence $\{n_q\}$, a sequence of positive definite matrices $D_q$ ($D_q$ should be a function with $(\theta,\delta)$, where we drop the argument there for notational simplicity), a random vector $Y$, and a positive definite correlation matrix $\Omega_0$, satisfying the consistency conditions:
\begin{eqnarray*}
  \sqrt{n_q}D_q^{-1/2}( \bar{p}_{n_q}(\theta_{n_q},\hat{\delta}_{n_q})-E[\bar{p}_{n_q}(\theta_{n_q},\delta_{n_q})] )
  &\xrightarrow{d}& Y \sim N(0,\Omega_0) \\
  D_q^{-1/2} \hat{\Sigma}_{n_q}(\theta_{n_q},\hat{\delta}_{n_q}) D_q^{-1/2}
  &\xrightarrow{p} & \Omega_0
\end{eqnarray*}

For simplicity, introduce the following notations:
\begin{eqnarray*}
  X &=& \Omega_0^{-1/2}Y \sim N(0,I) \\
  \hat{\Omega}_q(\theta_{n_q},\delta_{n_q},\hat{\delta}_{n_q})
                   &=& D_q^{-1/2}\hat{\Sigma}_{n_q}(\theta_{n_q},\hat{\delta}_{n_q})D_q^{-1/2} \\
  Y_q(\theta_{n_q},\delta_{n_q},\hat{\delta}_{n_q})
           &=& \sqrt{n_q}D_q^{-1/2}( \bar{p}_{n_q}(\theta_{n_q},\hat{\delta}_{n_q})-E[\bar{p}_{n_q}(\theta_{n_q},\delta_{n_q})] )  \\
  X_q(\theta_{n_q},\delta_{n_q},\hat{\delta}_{n_q})
        &=& \hat{\Omega}_q(\theta_{n_q},\delta_{n_q},\hat{\delta}_{n_q})^{-1/2} Y_q(\theta_{n_q},\delta_{n_q},\hat{\delta}_{n_q})
\end{eqnarray*}

Therefore, the assumption implies
\begin{eqnarray*}
  X_q &\rightarrow_d& X \sim N(0,I) \\
  \hat{\Omega}_q &\rightarrow_p& \Omega_0
\end{eqnarray*}

For each $q$, we define $\Lambda_q (\theta_{n_q},\delta_{n_q})$ as the same way as \cite{cox2023simple}, that is, a $d_A\times d_A$ diagonal matrix with positive entries on the diagonal such that each row of $\Lambda_q A(\theta_{n_q}) D_q^{1/2}$ is either zero or belongs to the unit circle. \cite{cox2023simple} argues such a $\Lambda_q$ always exists by taking the diagonal element to be the inverse of the magnitude of the corresponding row of $A(\theta_{n_q}) D_q$ if it is nonzero, and one otherwise.

Recall the test statistics for $n_q$
\begin{equation*}
  T_{n_q}(\theta_{n_q},\hat{\delta}_{n_q})=\min_{\kappa: A(\theta_{n_q}) \kappa \leq \rho(\theta_{n_q}) }
  n_q (\bar{p}_{n_q}(\theta_{n_q},\hat{\delta}_{n_q})-\kappa)' \hat{\Sigma}_{n_q}(\theta_{n_q},\hat{\delta}_{n_q})^{-1}(\bar{p}_{n_q}(\theta_{n_q},\hat{\delta}_{n_q})-\kappa)
\end{equation*}

Apply a change of variable $y=\sqrt{n_q}D_q^{-1/2}(\kappa-E[\bar{p}_{n_q}(\theta_{n_q},\delta_{n_q})])$.
Moreover, let
$g_q(\theta_{n_q},\delta_{n_q})=\sqrt{n_q}\Lambda_q( \rho(\theta_{n_q})-A(\theta_{n_q})E_{F_{n_q}}\bar{p}_{n_q}(\theta_{n_q},\delta_{n_q}) )$,
the test statistics can be written as
\begin{equation*}
  T_{n_q}(\theta_{n_q},\delta_{n_q},\hat{\delta}_{n_q})=\inf_{y:\Lambda_q A(\theta_{n_q})D_q^{1/2}y \leq g_q}
   (Y_q-y)'\hat{\Omega}_q^{-1}(Y_q-y)
\end{equation*}

Following \cite{cox2023simple}'s argument, there is
\begin{equation*}
  \mathds{1}\left\{ T_{n_q}(\theta_{n_q})\leq
  \chi^2_{\hat{r}(\theta_{n_q},\hat{\delta}_{n_q}),1-\hat{\beta}(\theta_{n_q},\hat{\delta}_{n_q})}
   \right\}  =
  \mathds{1}\left\{ \| x_q-t^*_q(x_q) \|^2 \leq
  \chi^2_{ r^q(x_q),1-\beta^q(x_q) }
   \right\}
\end{equation*}
where $\{x_q\}$ is a sequence of variables generated in their method, $\{t^*_q\}$, $\{r^q\}$, $\{\beta^q\}$ are sequences of functions derived in the same process. At the same time, \cite{cox2023simple} has proved the following expression
\begin{equation*}
  \mathds{1}\left\{ \| x_q-t^*_q(x_q) \|^2 \leq
  \chi^2_{ r^q(x_q),1-\beta^q(x_q) }
   \right\}
  \geq
  \mathds{1}\left\{ \| x_{\infty}-t^*_{\infty}(x_{\infty}) \|^2 \leq
  \chi^2_{ r^{\infty}(x_{\infty}),1-\beta^{\infty}(x_{\infty}) }
   \right\}
\end{equation*}
where $\{x_{\infty}\}$ is some variable and $\{t^*_{\infty}\}$, $\{r^{\infty}\}$, $\{\beta^{\infty}\}$ are some functions given by \cite{cox2023simple}.
Apply expectation to both side of this inequality, and by bounded convergence theorem, there is
\begin{equation*}
  \liminf_{q\rightarrow \infty} \text{Pr}_{F_{n_q}}(T_{n_q}(\theta_{n_q},\delta_{n_q})\leq \chi^2_{\hat{r}(\theta_{n_q},\delta_{n_q}),1-\hat{\beta}(\theta_{n_q},\delta_{n_q})}) \geq \text{Pr} (
  \| X-t^*_{\infty}(X,\Omega_0) \|^2 \leq \chi^2_{r^{\infty}(X,\Omega_0),1-\beta^{\infty}(X,\Omega_0)}        )
\end{equation*}

Moreover, Theorem \ref{th1} implies that
\begin{equation*}
  \text{Pr} (
  \| X-t^*_{\infty}(X,\Omega_0) \|^2 \leq \chi^2_{r^{\infty}(X,\Omega_0),1-\beta^{\infty}(X,\Omega_0)}        )  \geq 1-\alpha
\end{equation*}
under the situation $n=1$ and $\bar{p}_n(\theta,\delta,\hat{\delta})=X$. Therefore, these two inequalities imply the conclusion of Theorem \ref{th2}.
\end{proof}



\end{document}